\documentclass[a4paper,preprint]{elsarticle}

\PassOptionsToPackage{normalem}{ulem}
\usepackage{ulem}

\usepackage{algorithm}
\usepackage{algpseudocode}
\usepackage{graphicx}
\usepackage{epsfig}
\usepackage{caption}
\usepackage{subcaption}
\usepackage{amsmath}
\usepackage{xparse}
\usepackage{amssymb}
\usepackage{tcolorbox}
\usepackage{amsthm}
\usepackage{graphicx}

\usepackage{enumitem}

\newtheorem{definition}{Definition}

\newtheorem{lemma}{Lemma}
\newtheorem{theorem}{Theorem}

\newcommand{\calS}{\mathcal{S}}
\newcommand{\calP}{\mathcal{P}}
\newcommand{\calA}{\mathcal{A}}

\newcommand{\Thetaopt}{\Theta_{\rm OPT}}
\newcommand{\Thetakopt}{\Theta_{{\rm k-OPT}}}
\newcommand{\Ropt}{R_{\rm OPT}}
\newcommand{\Rmax}{R_{\max}}

\newcommand{\barr}{{\bar r}}

\newcommand{\emptyf}{f(\cdot,\,\cdot)}
\newcommand{\emptyr}{r(\cdot,\,\cdot)}
\newcommand{\emptybarr}{\bar r(\cdot,\,\cdot)}
\newcommand{\emptyrp}{r'(\cdot,\,\cdot)}
\newcommand{\emptytheta}{\Theta(\cdot,\,\cdot:s)}
\usepackage[colorinlistoftodos,textsize=tiny,textwidth=2cm,color=red!25!white,obeyFinal]{todonotes}

\title{Minmax-Regret $k$-Sink Location on a Dynamic Tree Network with Uniform Capacities}

\author[HKUST]{Mordecai Golin \fnref{fn1}}
\address[HKUST]{UMass Amherst.  mgolin@umass.edu}

\author[CMU]{Sai Sandeep \fnref{fn1}}
\address[CMU]{University of California, Berkeley. saisandeep192@gmail.com}
\fntext[fn1]{Work partially done while at the Hong Kong University of Science and Technology and partially supported by Hong Kong RGC CERG  Grant	16208415}
\date{}

\begin{document}

\begin{abstract}
A {\em dynamic flow network $G$}  with uniform {\em capacity} $c$ is a graph in which at most $c$ units of flow can enter an edge in one time unit. If flow enters a vertex faster than it can leave,  {\em congestion} occurs.

The evacuation problem is to evacuate all flow to  sinks. The $k$-sink location problem is to place $k$-sinks so as to minimize this evacuation time. 
A flow is {\em confluent} if all flow passing through a particular vertex must follow the same exit edge.  It is known that the confluent $1$-sink  location problem on a graph with $n$ nodes  is NP-Hard to approximate even within an $\Theta(\log n)$ factor.

The   $k$-sink  location problem restricted to trees, which partitions the tree into $k$ subtrees each containing  a sink, is polynomial solvable in  $\tilde O(k^2 n)$ time.

The concept of {\em minmax-regret} arises  from robust optimization.  
For each source, a range of possible flow values is provided and any scenario with flow values in those ranges might occur.  
The goal is to find a sink placement  that minimizes, over all possible scenarios, the difference between the evacuation time  to those sinks    and the minimal evacuation time of that scenario

Minmax-Regret $k$-Sink Location on  Dynamic {\em Path}  Networks with uniform capacities is polynomial solvable in $n$ and $k$.   Similarly, the Minmax-Regret $k$-{\em center} problem on trees is polynomial solvable in $n$ and $k$.  Prior to this work, polynomial time solutions for Minmax-Regret $k$-Sink Location on Dynamic Tree Networks with uniform capacities were only known for $k=1$.  This paper gives a 
 $$O\Bigl( \max(k^2 \log^2 k,\log ^2n)\,  k^4 n^2  \log^5 n\Bigr)$$ 
time solution to the problem.  The algorithm works for both the discrete case, in which  sinks are constrained to be vertices, and the continuous case, in which sinks may fall on edges as well.

\end{abstract}


\begin{keyword}
 Tree Partitioning \sep
Sink Evacuation    \sep
 Minmax-Regret   \sep
$k$-Center   \sep
Confluent Flows
\end{keyword}

\maketitle
\section{Introduction}
\label{sec:Intro}

{\em Dynamic flow networks}   were introduced by Ford and Fulkerson in 
\cite{Ford1958a}    to  model movement of items on a graph.  Each vertex in the graph is assigned some initial set of flow (supplies)  $w_v;$ if $w_v >0$ the vertex is a {\em source}. 
Each graph edge  $e=(u,v)$ has an associated  length $d(u,v)$,  which is the time  required to traverse the edge and a capacity $c_e,$ which is the  {\em rate} at which items can enter the edge. If  $c_e\equiv c$  for all edges $e,$ the network has {\em uniform capacity}.  A major difference between dynamic and static flows is that,  in dynamic flows, as flow moves around the graph, 
 {\em congestion} can occur as supplies back up at a vertex.

A large literature on such flows exist.
Good surveys of the problem and applications can be found in
\cite{Aronson1989,Fleischer2007,higashikawa2019survey,Skutella2009}. 
With only one source and one sink the problem of moving  flow as quickly as possible along one path from the source to the sink  is known as the {\em Quickest Path problem} and has  a long history \cite{Pascoal2006}.
A natural generalization  is the transshipment  problem, e.g., \cite{Hoppe2000b}, in which  the graph has several sources and sinks, with  supplies on  the sources and each sink having a specified demand.  The problem is  to find the quickest time required to satisfy all of the demands.  \cite{Hoppe2000b} provides a polynomial time algorithm for the transshipment  problem with later improvements by \cite{fleischer1998efficient}.

Dynamic Flows also model  {\em evacuation problems} \cite{Higashikawa2014}.  Vertices  represent rooms, flow represent people, edges represent hallways and sinks are exits out of the building.The  problem is to find a routing strategy (evacuation plan) that evacuates everyone  to the sinks   in minimum time. 
All flow passing through a vertex  is constrained to evacuate out through a single edge specified  by the plan (corresponding to a sign at that vertex stating ``this way out'').
Such a flow, in which  all flow through a vertex leaves through the same edge, is known as {\em confluent}\footnote{Confluent flows occur naturally in problems other than evacuations, e.g.,  packet forwarding and railway scheduling \cite{Dressler2010b}.}.  In general, confluent flows are difficult to construct \cite{Chen2007,Chen2006,Dressler2010b,Shepherd2015}.
   If P $\not=$ NP, then, even in a static graph,  it is  impossible to construct a constant-factor approximate optimal confluent flow in polynomial time   even with only one sink.

Returning to evacuation problems on dynamic flow graphs,  the basic optimization question is  to determine  a plan that minimizes the total time\footnote{There is a separate {\em minsum}  literature that  attempts to minimize the {\em average} time needed  for evacuation. In the same way that total time  can be viewed as a generalization of $k$-center,  average-time can be viewed as a generalization of $k$-median.  See, e.g., \cite{benkoczi2020minsum,bhattacharya2018n,higashikawa2021almost,manna2024minsum} for more details.} needed to evacuate all the people.   This differs from the transshipment problem in that even though sources have fixed supplies (the number of people  to be evacuated) sinks do not have fixed demands. They can accept as much flow as arrives there.  Note that  single-source single-sink confluent flow problem is exactly  the polynomially solvable   quickest path problem \cite{Pascoal2006} mentioned earlier.

Observe that if  edge capacities are  ``large enough'',   congestion can never occur and   flow starting at a vertex will always evacuate to its closest sink. In this case the {\em $k$-sink location problem} -- finding the location of $k$ sinks that minimize total evacuation time --reduces to the unweighted {\em  $k$-center problem.}  Although the unweighted $k$-center problem is 
 NP-Hard \cite[ND50]{garey1979computers} in $n $ {\em and} $k$ it is polynomially-time solvable for fixed $k.$  In contrast, Kamiyama {\em et al.}~\cite{Kamiyama} proves 
 by reduction to {\em Partition},   that, even the $1$-sink evacuation problem with the sink known in advance  is   NP-Hard for general graphs.
 By modifying similar results for static confluent graphs,  \cite{Golin2017sink}  extended this to show that even for  $k=1$ and the sink  location fixed in advance, it is still impossible to approximate the evacuation time to within a factor of $O(\log n)$ if P $\not=$ NP.

Research on finding {\em exact} quickest confluent dynamic flows is therefore restricted to special graphs, such as trees and paths.
 \cite{BGYKK2017} solves the $k$-sink location problem for  paths with uniform capacities  in
$\min\bigl(O(n + k^2   \log^2  n),\ O(n \log n)\bigr)$ time and  for paths with general capacities in 
$\min\bigl(O(n \log n + k^2  \log^4  n),\, O(n \log^3  n)\bigr)$ time.
 \cite{Mamada2006} gives an $O(n \log^2 n)$ algorithm for solving the
 $1$-sink problem on a dynamic tree network. \cite{Higashikawaa} improves this down to $O(n \log n)$ for uniform capacities. \cite{chengolin2016} gave an 
  $O(n k^2 \log^5 n)$ for the $k$-sink location problem on trees which they  later reduced down to  $O\Bigl( \max(k,\log n) k n \log^4 n\Bigr )$ time in \cite{chengolin2018}.
  These last two results were for general capacity edges.  They can both be reduced by a factor of $\log n$ for the uniform capacity version.

%
%

In {\em robust optimization},  the exact amount of flow located at a source is unknown at the time the evacuation plan is drawn up.  One approach to  robust optimization is  to assume that, for each source,  only  an (interval) {\em  range} within which that amount may fall is known. One method  to  deal with this type of uncertainty is to find a plan that minimizes the {\em regret}, e.g. the  maximum discrepancy between the evacuation time for the given plan on a fixed input and the plan that would give the minimum evacuation time for that particular input.  This is known as the {\em minmax-regret} problem. Minmax-regret optimization has been extensively studied for the
 $k$-median  \cite{Bhattacharya2013,conf/cocoon/BhattacharyaK12,Brodal2008,Yu2008}
and $k$-center problems \cite{Averbakh1997,Bhattacharya2014a,Puerto2008,Yu2008} 
 as well as many other optimization problems  \cite{Conde2008,Puerto2011,Ye2015}.
 \cite{aissi2009min,averbakh2004interval,Candia-Vejar2011,Kouvelis1997} provide an introduction to the literature.  Since most of these problems are NP-Hard to solve exactly on general graphs, the vast majority of the literature concentrates on  algorithms for  special graphs, in particular paths and trees.
 In particular, for later comparison, since the $k$-center problem is a special case of the $k$-sink location problem, we note that the minmax-regret $k$-center problem on trees can be solved in $O(n^2 \log^2 n \log \log n)$ time  \cite{Averbakh1997}.

  Recently there has been a series of new results for minmax-regret $k$-{\em sink} evacuation on special structure dynamic graphs, mostly with  uniform capacities. The $1$-sink minmax-regret problem on  a uniform capacity path was originally proposed by   \cite{ChengHKNSX13} who gave an $O(n \log^2 n)$ algorithm.
This was reduced down to $O(n \log n)$ by \cite{Higashikawa2014a,Wang2013,Wang2014b} and then to $O(n)$ by  \cite{Bhattacharya2015}. For $k=2$   \cite{Li2014}  gave  an $O(n^3 \log n)$ algorithm, later  reduced to $O(n \log^4 n)$ by
\cite{Bhattacharya2015}.  For general $k,$    \cite{AAGS14}  provides two algorithms. The first runs in   $O(nk^2 \log^{k+1}n)$ time, the second in     $O(n^3  \log n)$ time.
\cite{golin2022minmax}   describes how to solve the $k$-sink minmax-regret problem on  a {\em general}  capacity paths in  $O(n^4  \log n)$ time.

Xu and Li  solve the $1$-sink min-max regret problem on  a uniform capacity cycle in $O(n^3 \log n)$ time \cite{xu2015minimax}.  For   the $k$-sink min-max regret problem on  a uniform capacity tree  the only result known previously was for $k=1$.  
 \cite{Higashikawa2014} provides an $O(n \log^2 n)$ algorithm which was reduced to $O(n \log n)$ by \cite{Bhattacharya2015}.

No results for $k > 1$ were previously known.  This paper derives a
 $$O\Bigl( \max(k^2 \log^2 k,\log ^2n)\,  k^4 n^2  \log^5 n\Bigr)$$
algorithm for the problem.  We note that, similar to the $k$-center problem, there are two different variations of the $k$-sink location problem, a {\em discrete}  version and a {\em continuous} version (\cite{chen2015efficient} provides  a   discussion of the history). The discrete version  requires all sinks to be on vertices;  the continuous version permits sinks to be placed on edges as well. Our result holds for both versions.

Our algorithm will work by showing how to recast the minmax-regret $k$-sink location problem, which originally appears to be attempting to minimize a {\em global}  function of the tree, into a minmax {\em tree-partitioning} problem utilizing purely {\em local} functions on the subtrees.  It will then apply a new generic partitioning scheme developed in \cite{chengolin2016,chengolin2018}.\footnote{The scheme was introduced in   \cite{chengolin2016} but then generalized and extended to the continuous case in \cite{chengolin2018}. Going forward, we will therefore only reference \cite{chengolin2018}.}

Section  \ref{sec:minmax} introduces the tree partitioning framework of \cite{chengolin2018} and shows how sink evacuation fits into that framework.
Section \ref{sec:model} introduces the formalism of the regret problem.

Sections \ref{sec:wcs}  and \ref{sec:local-max} are the  new major technical contributions of this paper.
Section \ref{sec:wcs}  proves that, given a fixed partition of the input tree into subtrees, there are only a linear number of possible worst-case scenarios that achieve the minmax-regret for that partition.
Section \ref{sec:local-max} uses this fact to define a new {\em local} regret function on subtrees and then proves that solving the $k$-sink locaton problem using this new local regret function will solve the global regret problem.

Section \ref{sec:alg}  combines  the results of the previous sections, inserts them into the framework of \cite{chengolin2018} and shows that this immediately implies
\begin{theorem}
\label{thm:minmax}
The minmax-regret $k$-sink evacuation problem on trees can be solved in time
 $$O\Bigl( \max(k^2 \log^2 k,\log ^2n)\,  k^4 n^2  \log^5 n\Bigr).$$
This result holds for both the discrete and continuous  versions of the problem.
\end{theorem}
Note that for any fixed $k,$ this simplifies to $O(n^2 \log^7 n)$.

 We conclude by noting that Theorem \ref{thm:minmax}, similar to all the other results quoted on minmax-regret above with the exception of  \cite{golin2022minmax},  assumes {\em uniform} capacity.  This is because almost all results on minmax-regret have their own equivalent of Section \ref{sec:wcs}, proving that in their problem they only need to be concerned with a small number of worst-case scenarios. This ability to restrict scenarios has not been observed in the general capacity edge case for trees and thus there does not seem an obvious approach to attacking the minmax-regret problem for general capacity edges for trees.

\section{Minmax Monotone Functions on Trees}
\label{sec:minmax}
The following definitions have been modified\footnote{More specifically, the major modification is a change in notation to fit more cleanly into the framework of this paper.  For example, what are called   simply  ``partitions'' in  \cite{chengolin2018}   are called ``$k$-partitions'' here.  Also, conditions 1 and 2 on the next page were actually combined into one condition in the original in \cite{chengolin2018}.}
 from \cite{chengolin2018}.

\begin{definition} 
\label{def:Partitions}  (Fig.~\ref{fig:Prop 5})
Let $T=(V,E)$  be a tree.


\begin{enumerate}[label=\emph{\alph*})]

\item $\calP =\{P_1,P_2,\ldots,P_k\}$ is a {\em $k$-partition} of $V$ if  each subset $P_i\subseteq V$ induces a subtree, $\cup_i P_i = V$,  and $\forall i \not=j$,  $P_i \cap P_j=\emptyset$.\\
The $P_i$ will be the {\em blocks} of $\mathcal{P}$.

\item Let $X= \{x_1,x_2,\ldots,x_k \} \subseteq V.$  
$\Lambda[X]$ will denote the set of all $k$-partitions $\calP =\{P_1,P_2,\ldots,P_k\}$  of $V$ such that $\forall i,$   $ X \cap P_{i} = \{x_i\}$.\\ Depending upon the underlying problem, the $x_i$ are referred to as the  {\em centers} or {\em sinks} of the $P_i.$

\item For any subtree $T'=(V',E')$ and $x \in V'$, removing $x$ from $V'$ leaves a forest  $\mathcal{F} = \{ T_1,...,T_t \}$.  Let $V_1,\ldots,V_t$ denote the respective vertices in $T_i.$\\  $V_i$ will be a {\em branch} of $V'$ falling off of $x.$
\end{enumerate}
\end{definition}


Let $f : 2^{V} \times V \rightarrow [0,+\infty]$  be an \emph{atomic cost function}. If  $V' \subseteq V$,
$f(V',x)$ should be interpreted as 
{\em the cost for  $x$ to serve the set of nodes $V'$}. 
$f$ is {\em a minmax monotone function for one sink}  if it satisfies properties 1-5 below. 
\begin{enumerate}
\item {\em Tree Inclusion}\\
If $V'$ does not induce a subtree or $x \not\in V'$ then  $f(V',x) = \infty$;
\item {\em Nodes serve themselves for free.}\\
if $V' = \{x\}$, then $f(V',x) = 0$.

\item {\em Set Monotonicity} (larger sets cost more to serve) \\
If $x \in V'_1 \subseteq V'_2 \subseteq V$, then $f(V'_1,x) \leq f(V'_2,x).$

\item {\em Path Monotonicity} (moving sink away from tree increases cost)  \\Let $u \in V'$ and  $x \notin V'$ be a neighbor of $u$ in $T$.\\  Then $f(V' \cup \{ x \},x) \geq f(V',u)$. 

\item {\em Max Tree Composition} (Fig.~\ref{fig:Prop 5}) \\ Let $T' = (V',E')$ be a subtree of $T$ and $x \in V'$ a node with $t$ neighbors in $V'.$  
Let $V_1,...,V_t$ be  the branches of $T$ falling off of $x.$
 Then $$f(V',x) = \max_{1\leq i\leq t} f(V_i \cup \{ x\},x).$$

\begin{figure}[t]
	\centering
\includegraphics[width=2in]{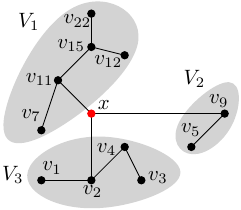}
\caption{Illustration of Definition \ref{def:Partitions} and Property 5. $V$ is the entire set of tree vertices.
The 3 branches of $V$ falling off of $x$ are $V_1,V_2,V_3$.  The cost of servicing branch $V_i$ is $f(V_i\cup \{x\},x)$.  The cost of servicing the entire tree using $x$  is $f(V,x) = \max_{1 \le i \le 3}f(V_i\cup \{x\},x).$}
	\label{fig:Prop 5}
\end{figure}

\end{enumerate}
Note  that 1-5 only define  a cost function over one subtree  and one single sink. Function $\emptyf$ is now naturally extended to work on  on partitions and sets.

\begin{enumerate}[resume]
\item  {\em Max Partition  Composition}  (Fig.~\ref{fig:Prop 6})  \\
\begin{equation}
\forall \mathcal{P} \in \Lambda[X],\quad   f(\mathcal{P},X) = \max_{1 \le i\le |X|} f(P_{i},x_i).
\label{eq:MaxCompositionB}
\end{equation}
\end{enumerate}

\begin{figure}[t]
	\centering
\includegraphics[width=5in]{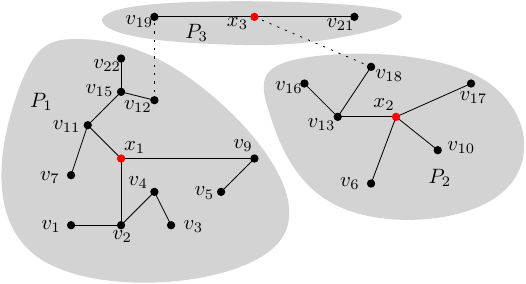}
\caption{Further illustration of Definition \ref{def:Partitions}  with Property 6. The complete tree has been partitioned into
$\calP=(P_1,P_2,P_3)$ with associated centers $X=(x_1,x_2,x_3).$  The cost associated with
$(\calP,X)$ is  $\max_{i=1,2,3} f(P_i,x_i)$. Note that the edges $(v_{12},v_{19})$ and $(x_3,v_{18})$ appear in the full tree but not in any of the $P_i$ subtrees.}
	\label{fig:Prop 6}
\end{figure}

\begin{definition}
\label{def:MMM}
A  cost function $\emptyf$  on $T$ that satisfies   properties 1-6 is called {\em minmax monotone}. 
\end{definition}

\begin{definition}
\label{def:kctpp}
Given a tree $T=(V,E)$ and  a minmax monotone cost-function $\emptyf$, the 
{\em minmax  $k$-center tree partitioning problem}, {\em $k$-center partitioning} for short,  is to find a set of $k$ sinks $X^* \subset V$ and a $k$-partition $\calP^* \in \Lambda[X]$ that minimize $f(\calP,X)$, i.e.,
$$f(\calP^*,X^*) = \min\Bigl\{f(\mathcal{P},X)\,:\,  X \subset V,\, |X| = k,\, \mathcal{P} \in \Lambda[X]\Bigr\}.
$$
\end{definition}
\par\noindent {\small \em Note: For shorthand, we sometimes write
$$f(\calP^*,X^*) =
 \min_{(\calP,X):  |X| = k}f(\calP,X).
$$
This is valid because,  if  $\mathcal{P} \not  \in \Lambda[X],$ then, from property 1, $f(\calP,X)= \infty).$
}

 \cite{chengolin2018}  describes  a generic technique for solving this  problem given an oracle for calculating the  cost of a one sink solution given the sink location.

\begin{definition}\label{def:Orac2}
Let $T=(V,E)$ be the input tree.

$\calA$ is an {\em oracle} for $\emptyf$ if, for all subtrees $T'=(V',E')$    of $T$ and $x \in V'$,  $\calA$ calculates  $f(V',x).$

$t_{\calA}(n')$ denotes the worst case time required for $\calA$ to calculate $f(V',x)$ for a subtree of size $n'=|V'|.$ $\calA$ is 
{\em asymptotically subadditive} if
\begin{itemize}
\item   $t_\mathcal{A}(n') = \Omega(n')$ and is non-decreasing.
\item For all nonnegative $n_i$,  $\sum_i t_\mathcal{A}(n_i)   = O \left( t_\mathcal{A}\Bigl( \sum_i n_i \Bigr) \right).$
\item $t_\mathcal{A}(n'+1) = O\left(t_\mathcal{A}(n')\right)$.
\end{itemize}
\end{definition}
Note that for $x\ge 1$ and  $y\ge 0$ any function of the form  $n^x \log^y n$ is asymptotically subadditive.

\begin{theorem}[\cite{chengolin2018}]
	\label{Thm:Di}
Let $\emptyf$ be a monotone minmax function and 
$\mathcal{A}$  an asymptotically subadditive oracle for $\emptyf.$
Then  the  $k$-center partitioning problem on  $T$ can be solved in  time $$O\bigl( \max(k\log k,\log n) k^2 t_{\mathcal{A}}(n) \log^2 n \bigr).$$
	 \label{theorem:FastC}
\end{theorem}
The algorithm underlying this Theorem is based on  parametric search and its  correctness only depends upon the properties
listed in Definitions  \ref{def:Partitions}-\ref{def:Orac2}.

For context, we note that if $f(V',x) =\max_{v \in V'} w_v d(v,x)$
where $w_v$ is the weight on node $v$ and $d(v,x)$ is the length of the unique  from $v$ to $x$ then the $k$-center partitioning problem is just the classic weighted $k$-center on a tree problem.
$f(\cdot,\cdot)$ is  a monotone minmax function.
Since  $f(V',x)$  can be evaluated in time $O(|V'|)$,  $t_{\mathcal{A}}(n)=O(n)$, and is asymptotically subadditive, so  Theorem \ref {Thm:Di}  immediately implies that the weighted $k$-center problem can be solved in time
$$O\bigl( \max(k\log k,\log n) k^2 n \log^2 n \bigr).$$

We emphasize that this is NOT an interesting result on its own since weighted $k$-center on a tree can  already be solved in $O( n \log n)$ time \cite{wang2021n}. It is only intended to illustrate a simple application of Theorem \ref {Thm:Di}.

\subsection{Dynamic Confluent Flows on Trees -- Evacuation protocols}

The formal input to this problem is a tree $T=(V,E)$ along with
\begin{itemize}
\item  A {\em scenario}  $s=(w_{v_1}(s), w_{v_2}(s),\ldots, w_{v_n}(s))$.\\
The problem starts with  
$w_{v_i}(s) \ge 0$ ``units''  of flow items  located  on vertex $v_i \in V$.  All of this flow needs to be evacuated (moved)  to some sink.

\item For every  edge $e = (u,v) \in E$,  an associated length $d(u,v)>0$, denoting the time it takes a particle of flow to travel from $u$ to $v$. For $(u,v) \not\in E,$  $d(u,v)$ is defined to be the sum of the lengths of the edges on the unique path in $T$ connecting $u$ and $v.$
\item A {\em capacity} $c > 0$, denoting the {\em rate}, i.e., amount of flow  that can enter an edge in one unit of time.
\end{itemize}
To move $w$ units of flow along edge $e=(u,v)$ from $u$ to $v$  note that the last particle of flow requires waiting   $w/c$ units of time to enter $e$. After travelling along $e$ it finally arrives at $v$ at time $  \frac {w} c +  d(u,v) .$

Let $T'=(V',E')$ be some subtree of $T$   with $x \in V'.$ 
$\Theta(V',x:s)$  denotes the time required to evacuate all $w_{v'}(s)$  flow  on all $v' \in V'$ to sink $x$. This is the {\em last} time at which a particle of flow reaches $x.$   

{\em Congestion} occurs when too many items are waiting to enter an edge.  Congestion can build up if items arrive at a vertex  faster than the rate ($c$) at which they leave it.  This can happen if multiple edges feed into one vertex.
The formula for $\Theta(V',x:s)$ has been derived by multiple authors, e.g., \cite{Higashikawaa,Mamada2006} in different ways.  
Definition \ref{def:DW} and Lemma \ref{lem:unit flow} below are taken, slightly  modified\footnote{The modification is just a change of notation to fit into the framework of this paper.} from  \cite[2.2.5]{Higashikawa2014c}.
\begin{definition}
\label{def:DW}
Let $T'=(V',E')$ be a subtree and  $x \not\in V'$  be a neighbor of some node in $V'$.
For every $v' \in V',$  define
\begin{equation}
\label{eq:Ddef}
D(V',x,v') =
\{v \in V' \,:\, d(v,x) \ge d(v',x) \}
\end{equation}
Further set
\begin{equation}
\label{eq:Wdef}
W(V',x,v':s) = \sum_{v \in D(V',x,v')} w_v(s).
\end{equation}
\end{definition}

%

\begin{lemma} \label{lem:unit flow}\  \\
 Let $T'=(V',E') $ be  a subtree of $T$ and $x \not\in V'$ be  a neighbor of some node in $V'$.   Then
\begin{equation}
\label{eq:naf}
\Theta(V'  \cup \{x\}, x: s) =
\max_{ v' \in V'} 
\left(
d(v',x) +  \frac {W(V',x,v':s)} c  
\right)
\end{equation}
\end{lemma}

\begin{figure}[t]
{\hspace{-0.7in} \includegraphics[width=6in]{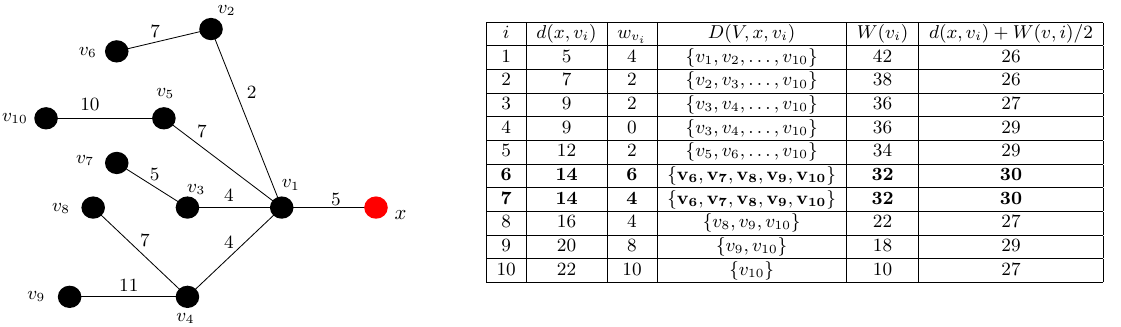}}
\caption{Example of evacuation to sink $x.$ All edges have capacity $c=2.$ 
Edge lengths are the number above the edge.  
Vertex indices are sorted by increasing distance from $x$ with ties being broken arbitrarily.  Vertex weights $w_v$ are not on the graph but are shown in the accompanying table. 
$W(v_i) = \sum_{v \in D(V,x,v_i)} w(v).$ 
The rightmost column shows the values compared in the right hand side of (\ref{eq:naf}).  Note that the cost of the tree is 30, which is when the last item (from $v_{10}$) will reach $x.$ In (\ref{eq:naf}),  this is achieved when $v= v_6$ or $v=v_7.$}
\label{fig:Evac}
\end{figure}

The function $\emptytheta$  trivially satisfies properties  1-4 in the previous section.  To see that it also satisfies property 5, let 
 $V_1,V_2,\ldots, V_t$ be the branches of $V'$ falling off of $x.$ Since $x$ has a unique neighbor 
$u_i \in V_i$, all items in $V_i$ evacuating to $x$ must pass through edge $(u_i,x)$ and do not interact at all with items from any other subtree $V_j,$  $j \not=i$ that are also evacuating to $x.$  Thus
$$ \Theta(V',x:s) = \max_{1 \le i\le t}  (\Theta(V_i \cup \{x\}, x:s))$$
and $\Theta(V',x:s)$  satisfies property 5. It is now extended  naturally to work on partitions and sets.
\begin{definition}
For any set of $k$ sinks $X$ and partition $\calP \in \Lambda[X]$ define
$$
\Theta(\calP,X: s) =
                              \max\{\Theta(P_1, x_1:s),  \Theta(P_2, x_2:s), \ldots, \Theta(P_k, x_k:s) \}.
$$

Further define
$$\Thetakopt(s) = \min \bigl\{   \Theta(\calP,X: s) \,:\,  |X| =k  \mbox{ and }  \calP \in \Lambda[X]\bigr\}$$
to be the minimum time required to evacuate all items.  The $(\calP,X)$ pair achieving this value is an {\em optimal evacuation protocol}.  When $k$ is fixed and understood we will write $ \Thetaopt$ instead of $\Thetakopt$.
\end{definition}

Intuitively, $(\calP,X)$ denotes that $T$  is partitioned into subtrees, each containing one sink to which all flow in the subtree evacuates. $\Theta(\calP,X: s) $ is the time required to evacuate all of the items with $(\calP,X)$ under scenario $s.$  $\Thetakopt(s)$ is the minimum time required to evacuate the entire tree if it is $k$-partitioned (and all edges have uniform capacity). This was solved in $O(n \log n)$ time for $k=1.$

Lemma \ref{lem:unit flow}  gives an  immediate   $O(n' \log n')$  oracle   for solving the rooted one-sink version of the problem, i.e.,  when the sink is known in advance.   This is to use an $O(n')$ breadth first search starting at $x$ to separate $V'$ into its branches while calculating all of the $d(v',x)$ values.  For each branch $V_i$, sort the $d(v',x)$  for $v' \in V_i$ by increasing value  in $O(|V_i| \log |V_i|)$  time   and calculate 
$\Theta(V'  \cup \{x\}, x: s)$ in $O(|V_i|)$ time using brute force.  Then return  the maximum over all of the branch values.  By plugging this 
$ t_{\mathcal{A}}(n') = O(n' \log n')$ oracle into 
 Theorem  \ref{theorem:FastC}, \cite{chengolin2018} derived    
 
 \begin{theorem} 
\label{thm: ksink}\cite{chengolin2018}
There is an  algorithm that solves the $k$-sink location problem on trees with uniform capacity in 
 $$O( \max(k\log k,\log n)\,  k^2 n  \log^3 n)$$ time.
 \end{theorem}


\section{Regret}
\label{sec:model}

In a minmax-{\em regret} model on trees,  the input tree $T$ is given  but some of the other input values are not fully known in advance.  The input specifies   restrictions on the allowed values for the missing inputs.

Concretely, in the  minmax-regret $k$-sink evacuation problem on trees the input tree $T$, capacity $c$ and  lengths $d(u,v)$ are all explicitly specified  as part of the  input.
The weights $w_v$ are {\em not}  fully specified in advance.  Instead, for each $v \in V$,  a  {\em range}  $[w_{i}^{-},w_{i}^{+}]$ within which $w_i$ must lie  is specified. $0 < w_v^+  \le w_v^-.$
The set of all possible allowed  scenarios is  the Cartesian product of all weight intervals,
\[ \calS=\prod_{v \in V}[w_{v}^{-},w_{v}^{+}],
\]  
 $s\in\mathcal{S}$ is an assignment of weights to all vertices. The {\em weight of a vertex $v$ under scenario $s$} is denoted by $w_{v}(s)$.

\begin{figure}[t]
	\centering
\hspace*{-.5in}\includegraphics[width=6.3in]{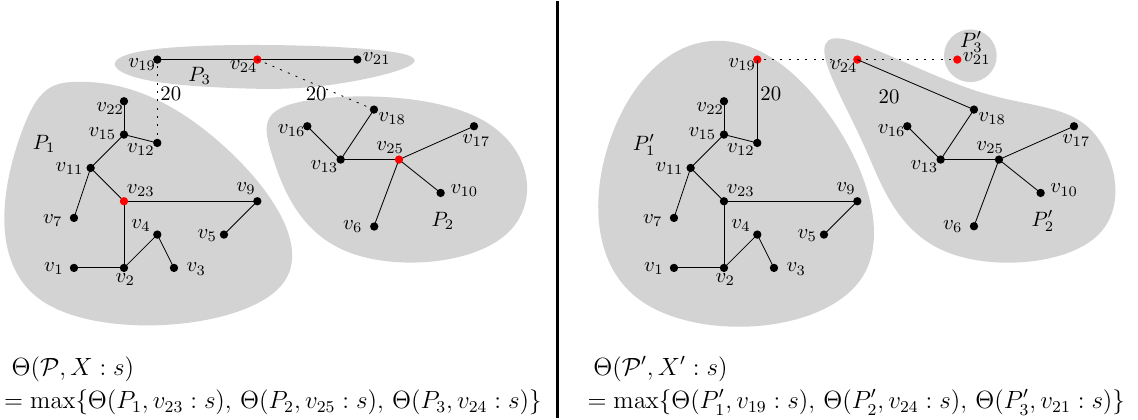}
\caption{Illustration of regret.  Capacity $c=1$.  All edges have length $d(u,v) =1$  except for 
$d(v_{12},v_{19}) = d(v_{18},v_{24}) = 20.$   Two different partitions $(\calP,X)$  and $(\calP',X')$ are illustrated on the left and right.  $X = \{v_{23}, v_{25}, v_{24}\}$ and $X' = \{v_{19}, v_{24}, v_{21}\}$.
For all vertices $v_i$ except for $i=19,21,24$,  
$[w^-_{v_i},w^+_{v_i}]=[1,2]$.  For  $i=19,21,24$,  $[w^-_{v_i},w^+_{v_i}]=[1,100]$. Let $s^*$ be a scenario for which $w_{v}(s^*)  = 1$ for all $v$ except for $w_{v_{19}}(s^*) = w_{v_{21}}(s^*)  = w_{v_{24}}(s^*) =100.$  Then  $\Theta(\calP,X:s^*) = \Theta(P_3,v_{24}:s^*) = 101.$  A little calculation shows that $\Theta(\calP',X':s^*) = \Theta(P_1,v_{19}:s^*)=30 = \Thetaopt(s^*). $ So $R(\calP,X : s) = 101-30=71.$
It is not that difficult to show that this $s^*$ is a worst case scenario  for  $(\calP,X)$ so $ \Rmax(\calP,X) = 71.$
}
	\label{fig:Regret1}
	\end{figure}

\begin{definition}[Regret for $(\calP,X)$ under scenario $s$] \ \\
For fixed   
$(\calP,X)$  with $|X| =k$ and   $s\in\cal{S}$, the {\em regret} is defined as the difference between
 $\Theta(\calP,X: s)$ and  the optimal $k$-Sink evacuation time for $s,$ i.e., 
\begin{equation} R(\calP,X : s) =  \Theta(\calP,X: s)  - \Thetaopt(s).  \label{eq:regret}
\end{equation}
\end{definition}

\begin{definition}[Max-Regret for $(P,X)$] \ \\
\label{def:regret}
The {\em Maximum-Regret}  achieved (over all scenarios) for a choice of $(P,X)$ is
\begin{equation}   \Rmax(\calP,X)= \max_{s\in\calS}\left\{ R(\calP,X:s)  \right\}. \label{eq:max-regret}
\end{equation}
$s^* \in \calS $ is  a \emph{worst-case scenario} for $(P,X)$  if $ R_{\max}(\calP,X)=R(\calP,X:s^*).$ 
\end{definition}

Finally, set 
\begin{definition}[Global Minmax-Regret] \ \\
Let $k$ be fixed. 
$$ \Ropt = \min_{(\calP,X):  |X| = k}R_{\max}(\calP,X).
$$
$(\calP^*,X^*)$ is an {\em optimal minmax-regret evacuation protocol} if  $\Ropt = R_{\max}(\calP^*,X^*)$.
\end{definition}

For later use, we note that, by definition,   $\Rmax(\calP,X) \ge 0$, and thus
\begin{lemma}
\label{lem:pos}
$$ \forall X \subseteq V  \mbox{\rm  and }   \calP \in \Lambda[X],\quad \Rmax(\calP,X) \ge 0.$$
\end{lemma}

\subsection{Minmax-Regret $k$-Sink Location Problem}
As described above, the input for the Minmax-Regret $k$-Sink Location Problem  is a dynamic flow  tree network  with   edge lengths, vertex weight {\em intervals} $[w_i^-,w_i^+]$ and  edge capacity $c$.
The goal is to calculate $\Ropt$ along with a corresponding  optimal minmax-regret evacuation protocol $(\calP^*,X^*)$ with associated worst case scenario $s^*.$

Intuitively, the goal  in the $k$-sink location problem is to find the best (optimal) sink location(s) that can evacuate all of the items as quickly as possible.  The goal in the  minmax regret version is to find the best location(s) given only partial information.  ``Best'' is now defined as the solution with smallest worst case difference from optimal under all allowable scenarios.

 This setup can also be viewed  as a 2-person Stackelberg game 
 between the algorithm $A$ and  adversary $B$:

\begin{enumerate}
  \item Algorithm $A$ (leader): creates an evacuation protocol  $(\calP,X)$.
  \item Adversary $B$ (follower): chooses a  \emph{worst-case scenario} $s^*\in\mathcal{S}$ for  $(\calP,X)$    i.e.,   $R(\calP,X:s^*)  = R_{\max}(\calP,X)$.
\end{enumerate}

$A$'s objective is to minimize the value of  $R(\calP,X:s^*)$ which is equivalent to finding  optimal minmax-regret evacuation protocol $(\calP^*,X^*).$

Note that even though we have defined regret only for the $k$-sink evacuation problem, this formulation can be (and has been) extended to many other problems by replacing $\Theta$ with some  other  function.

\medskip

\section{Worst Case Scenario Properties}
\label{sec:wcs}

The explicit formula for evacuation time given by Eq.~(\ref{eq:naf})
immediately implies that increasing weights cannot decrease evacuation time, written as: 
\begin{lemma}
\label{lem:dom2}
Let $s, s' \in \calS$ be two scenarios such that $\forall v \in V,  w_v(s') \le w_v(s)$.  
\begin{enumerate}
\item  If $T'=(V',E') $ is  a subtree of $T$   and $x \in V'.$ 
$$ \Theta(V', x: s') \le   \Theta(V', x: s).$$
\item $\Thetaopt(s') \le \Thetaopt(s).$
\end{enumerate}
\end{lemma}

 Eq.~\ref{eq:naf} 
 also immediately implies 
\begin{lemma}
\label{lem:dom3}
Let $s$ be a scenario and $s'$ be another scenario such that, for some $v' \in V$, and some $\delta > 0$ 
$$w_v(s')
=
\left\{
\begin{array}{cc}
w_v(s) & \mbox{if $v \not= v'$},\\
w_{v}(s) + \delta & \mbox{if $v = v'$}.
\end{array}
\right.
$$
Then 
$$\Thetaopt(s')  \le \Thetaopt(s) + \frac{\delta}{c}.$$
\end{lemma}

\begin{definition}[Dominant Subtrees and Branches]\ \\
\label{def:dom}
$P_i \in \calP$ is  a {\em dominant subtree for $(\calP,X)$ under scenario $s$ } if
$$ \Theta(\calP,X:s) = \Theta(P_i,x_i:s).$$
For any $P_i$, a  {\em dominant branch of $P_i$ falling off of $x_i$} is a branch $V_j$ of $P_i$ falling off of $x_i$ such that 
$$\Theta(P_i,x_i:s) = \Theta(V_j \cup\{x_i\},x_i:s).$$
\end{definition}
Note that, by this definition, if
 $s^*\in \calS$ is a worst-case scenario  for $(\calP,X)$ and $P_i \in \calP$ is  a dominant subtree for $(\calP,X)$ under scenario $s^*$ then 
$$\Rmax(\calP,X) = \Theta(\calP,X: s^*)  - \Thetaopt(s^*) = \Theta(P_i,x_i: s^*)  - \Thetaopt(s^*).$$
Furthermore, if $V_j$ is a dominant branch of  that $P_i$ falling off of $x_i$ under $s$ then
$$\Rmax(\calP,X) = \Theta(P_i,x_i: s^*)  - \Thetaopt(s^*) = \Theta(V_j \cup\{x_i\},x_i:s) - \Thetaopt(s^*).$$

\begin{figure}[t]
	\centering
\includegraphics[width=3in]{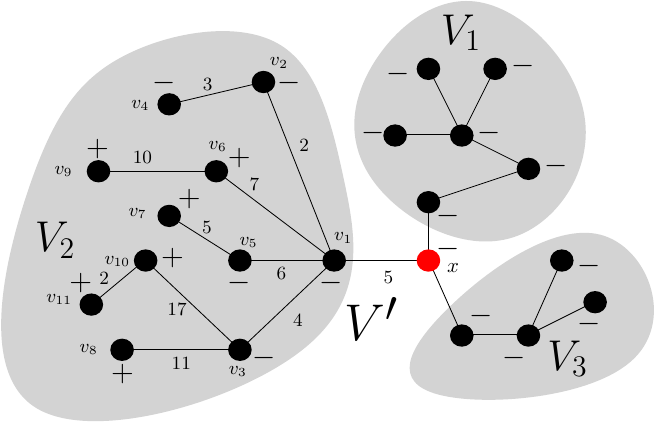}
\caption{Illustration of Definition  \ref {def:S*tree}.   The edge lengths are only explicitly given for the subtree $V_2$. For exposition, the $v_i$ in $V_2$ are labelled so that $i$ increases as $d(v_i,x)$ increases.   
The scenario $s^*=s^*(V_2,x,v_6)$ is illustrated.  Note that 
$D(V_2,x,v_6)=\{v_6,v_7,v_8,v_9,v_{10},v_{11}\}$, i.e., the set of vertices $v' \in V_2$ such that $d(x,v')\ge d(x,v_6) = 12.$
The vertices $v' \in D(V_2,x,v_6) $ are labelled with a ``+'  to signify that    $s^*(`v) = w^+_{v'}.$  The  vertices  $v \in V\setminus  D(V_2,x,v_6)$ are  labelled with a ``-'' to signify that   $s^*(v) = w^-_v.$
}
	\label{fig:DV}
	\end{figure}

The major observation will be that there 
is  always  a ``worst-case'' scenario in which the weight of each vertex is either at the lower or the upper end  of its range and only a small number of such scenarios need to be considered. Proving this  requires introducing new notation.

\begin{definition}\ \\
\label{def:S*tree}
\begin{enumerate}
\item   Let $V' \subseteq V$ ($V'$ is not necessarily a tree). Set
$s^*(V')$  to be the scenario such that
$$
w_v(s^*(V'))
=
\left\{
\begin{array}{cc}
w^-_v  & \mbox{if  $v \not\in V'$}\\
w^+_v  & \mbox{if $v \in V'.$}
\end{array}
\right.
$$
Note that,  $\forall v \in V,$  $w_v(s^*(\emptyset)) = w^-_v.$
\item  Let $T'=(V',E')$ be a subtree with $x \in V'$ and  $V_1,...,V_t$  be  the branches of $V'$ falling off of $x.$  Now set (Fig.~\ref{fig:DV})
\begin{equation}\label{eq:scendefs}
\calS^*(V',x) = s^*(\emptyset)   \cup \bigcup_{j=1}^t  \left\{  s^*(D(V_j,x,v')) \,:\, v' \in V_j  \right\}
\end{equation}
where  $D(V_j,x,v')$ is as introduced in Definition \ref{def:DW}.
%
\item   Set 
$$\calS^*(\calP,X) = \bigcup_{i=1}^k \calS^*(P_i,x_i).$$
\end{enumerate}
\end{definition}

Note that $\calS^*(V',x) $ contains at most one scenario associated with each $v \in V'$ (associate  $s^*(\emptyset)$ with $x$)  and thus $|\calS^*(V',x) | \le |V'|$ and $|\calS^*(\calP,X)| \le n.$
The main result is

\begin{lemma}
\label{lem:wcr}
Let $s \in \calS$ be a worst case scenario for $(\calP,X)$  and $P_i \in \calP$ be   a dominant subtree for $(\calP,X)$ under scenario $s$. 
Furthermore, let  $V_j$ be  a dominant branch in $P_i$ falling off of $x_i$ under $s$.

Then there exists 
$s^* \in \calS^*(V_j \cup \{x_i\},x_i) \subseteq  \calS^*(P_i,x_i) \subseteq \calS^*(\calP,X)$ satisfying
\begin{enumerate}
\item   $s^*$  is also a  worst case scenario for  $(\calP,X)$ and 
\item $P_i \in \calP$ is  a dominant subtree for $(\calP,X)$ under scenario $s^*$,
\item with  $V_j$ being  a dominant branch of $P_i$ falling off of $x_i$ under $s^*.$
\end{enumerate}
Furthermore
\begin{eqnarray}
\label{eq:rmresult}
\Rmax(\calP,X) = \Theta(\calP,X: s^*)  - \Thetaopt(s^*) &=& \Theta(P_i,x_i: s^*)  - \Thetaopt(s^*)\\
&=& \Theta(V_j\cup\{x_i\},x_i: s^*)  - \Thetaopt(s^*). \nonumber
\end{eqnarray}
\end{lemma}
\begin{proof}
By definition, $\calP, X, s$ satisfy
$$\Rmax(\calP,X) = \Theta(P_{i},x_i:s) - \Thetaopt(s).$$

Now let $V_1,\ldots,V_t$ be the branches of $P_i$ hanging off of $x_i$ and let $V_j$ be a dominant branch of $P_i$ under $s.$ Then
$$\Theta(P_i,x_i:s)  = \max_{1 \le \ell \le t} \Theta(V_{\ell} \cup\{x_i\},x_i:s)=  \Theta(V_{j}\cup\{x_i\},x_i:s) $$

Recall from Lemma \ref{lem:unit flow}\  that
$$\Theta(V_j  \cup \{x_i\}, x_i: s') =
\max_{ v' \in V_j} 
\left(
d(v',x_i) +  \frac {W(V_j,x_i,v':s')} c  
\right).
$$
Let $ v^*\in V_j$ be any vertex (if ties occur, there might be many) 
 such that
\begin{equation}
\label{eq:barv def}\Theta(V_j  \cup \{x_i\}, x_i: s') =
d(v^*,x_i) +  \frac {W(V_j,x_i, v^*:s')} c .
\end{equation}

We claim that  $s^*=s^*(D(V_j,x_i,v^*))$ is a worst-case scenario for $(\calP,X)$ that satisfies conditions (1), (2) and (3). The correctness of this claim would immediately imply the lemma.  The proof of the claim follows in three parts.  

\begin{itemize}
\item[(a)] For all $v \not\in V_j$,  it will set $w_v(s) = w_v^-.$
\item[(b)] For all $v \in V_j$ with  $d(v,x_i) < d(v,x^*) $, it will set $w_v(s) = w_v^-.$
\item[(c)] For all $v \in V_j$ with  $d(v,x_i) \ge  d(v,x^*) $, it will set $w_v(s) = w_v^+.$
\end{itemize}
Each step will be shown to maintain  $s$ as  a worst-case scenario for $(\calP,X)$ that satisfies conditions (1), (2) and (3).  After all the modifications are made the resulting scenario is  $s^*=s^*(D(V_j,x_i,v^*))$, completing the proof of the claim.

\medskip

 \par\noindent\underline{(a) Reducing $w$ values outside of dominant branch:}\\
Define $s_a$ such that 
$$
\forall v \in V,\quad
w_v(s_a) = \left\{
\begin{array}{cc}
w^-_v & \mbox{if $v \not\in V_{j}$,}\\
 w_v(s)& \mbox{if $v\in V_{j}$.}  
\end{array}
\right.
$$

From Lemma \ref{lem:dom2} (1),  
$$\forall i' \not=i,\quad\Theta(P_{i'},x_{i'}:s_a) \le \Theta(P_{i'},x_{i'}:s)$$
and
$$\forall \ell \not=j,\quad \Theta(V_{\ell}\cup \{x_i\},x_{i}:s_a) \le \Theta(V_{\ell}\cup \{x_i\},x_{i}:s).
$$
Since $s$ and $s_a$ are identical within $V_{j},$   
$$\Theta(V_{j}\cup \{x_i\},x_{i}:s_a) = \Theta(V_{j}\cup \{x_i\},x_{i}:s).$$
Thus
$$\Theta(P_i,x_i:s_a)  = \max_{1 \le \ell \le t} \Theta(V_{\ell} \cup\{x_i\},x_i:s_a)=  \Theta(V_{j}\cup\{x_{i}\},x_i:s) = \Theta(P_i,x_i:s)$$
and
$$\Theta(\calP,X: s_a) = \max_{1 \le i' \le k}   \Theta(P_{i'},x_{i'}:s_a) = \Theta(P_i,x_i:s_a) = \Theta(P_i,x_i:s) = \Theta(\calP,X: s).$$

From Lemma \ref{lem:dom2} (2),  
$$\Thetaopt(s_a) \le \Thetaopt(s)$$ and thus
$$\Theta(\calP,X: s_a) - \Thetaopt(s_a)  \ge \Theta(\calP,X: s)   - \Thetaopt(s)  = \Rmax(\calP,X).$$
From Definition \ref{def:regret}, 
$\Theta(\calP,X: s_a) - \Thetaopt(s_a) \le \Rmax(\calP,X)$ and thus
$$\Rmax(\calP,X) = \Theta(P_{i},x_i:s_a) - \Thetaopt(s_a).$$

Thus  $s_a$ is also a worst case scenario for $(\calP,X)$  with  $P_i \in \calP$  a dominant subtree for $(\calP,X)$ under scenario $s_a$, with $V_j$ a dominant branch in $P_i$.

Before continuing to (b), set $s'=s_a.$

\medskip

\par\noindent\underline{(b) Reducing $w$ values inside  of dominant branch:}\\

Suppose there now exists some 
  $v' \in V_j$ such that 
  $d(v',x_i) < d(v^*, x_i)$ and  $w_{v'}(s') > w_v^-$. 
 In this case, transform  $s'$ into $s''$ by setting 
 \begin{equation}
 \label{eq:wchange}
 \forall v \in V,\quad 
 w_v(s'') = \left\{
\begin{array}{cc}
w^-_v & \mbox{if $v = v'$,}\\
 w_{v}(s')& \mbox{if $v \not= v'$.}  
\end{array}
\right.
\end{equation}
%
We now show that $s''$ will remain  a worst case scenario for   $(\calP,X)$ satisfying conditions (1), (2) and (3).

Because $w_v(s'') = w_v(s')$ for all $v \not\in V_j,$ we have
$$\forall i' \not=i,\quad\Theta(P_{i'},x_{i'}:s'') = \Theta(P_{i'},x_{i'}:s')
$$
and
$$
\forall \ell \not=j,\quad \Theta(V_{\ell}\cup\{x_i\},x_{i}:s'')  = \Theta(V_{\ell}\cup\{x_i\},x_{i}:s').
$$
Now note that, for all $v \in V_j$, the weight change in (\ref{eq:wchange})  implies 
$$W(V_j,x_i,v:s'') =
\left\{
\begin{array}{ll}
W(V_j,x_i,v:s') - (w_{v'}(s') - w^-_{v'})  & \mbox{if $d(v,x_i) \le d(v',x_i)$},\\
W(V_j,x_i,v:s') 					&  \mbox{if $d(v,x_i) > d(v',x_i)$}.
\end{array}
\right.
$$
Since  $d(v^*,x_i) > d(v',x_i)$,  $W(V_j,x_i,v^*:s'')=W(V_j,x_i,v^*:s')$ and 
\begin{eqnarray*}
\Theta(V_j  \cup \{x_i\}, x_i: s'') &=&
\max_{ v \in V_j} 
\left(
d(v,x_i) +  \frac {W(V_j,x_i,v:s'')} c  
\right)\\
&=&d(v^*,x_i) +  \frac {W(V_j, x_i, v^* :s'')} c\\
&=&d(v^*,x_i) +  \frac {W(V_j, x_i, v^*:s')} c\\
&=& \Theta(V_j  \cup \{x_i\}, x_i: s').
\end{eqnarray*}

Thus, exactly as in (a), 
$$\Theta(P_i,x_i:s'')  = \max_{1 \le \ell \le t} \Theta(V_{\ell} \cup\{x_i\},x_i:s'')=  \Theta(V_{j}\cup\{x_{i}\},x_i:s') = \Theta(P_i,x_i:s')$$
and
$$\Theta(\calP,X: s'') = \Theta(P_i,x_i:s'')  = \Theta(P_i,x_i:s') = \Theta(\calP,X: s').$$

Again from Lemma \ref{lem:dom2} (2),  
$$\Thetaopt(s'') \le \Thetaopt(s')$$ and using the same argument as in (a),
$$\Rmax(\calP,X) = \Theta(P_{i},x_i:s'') - \Thetaopt(s'').$$
Thus   $s''$ is also a worst case scenario for $(\calP,X)$  with  $P_i \in \calP$  a dominant subtree for $(\calP,X)$ under scenario $s''$, with $V_j$ a dominant branch in $P_i$, i.e., satisfying conditions (1), (2) and (3).

Now, set $s'=s''$. 
As long as there exists  some  $v' \in V_j$ such that 
  $d(v',x_i) < d(v^*, x_i)$ and  $w_{v'}(s') > w_v^-$,
the argument above can be repeated.  Since this process only reduces weights and never increases them, it must terminate after a finite number of steps.
Let $s_b$ be the scenario after the process terminates.

By construction,  $s_b$ is  a worst case scenario for $(\calP,X)$ satisfying conditions (1), (2) and (3).  Furthermore
if  $ v \not\in V_j$  or  $v \in V_j$ but $d(v,x_i) < d(v*,x_i)$,  then $w_v(s_b) = w_v^-.$

Before continuing to (c), set $s'=s_b.$
\medskip

\par\noindent\underline{(c) Increasing  $w$ values inside  of dominant branch:}\\

Suppose  there now exists some  $v' \in V_j$ such that 
  $d(v',x_i) \ge  d(v^*, x_i)$ and  $w_{v'}(s') < w_v^+$, 
Then  transform  $s'$ into $s''$ by setting 
 $$
 \forall v \in V,\quad
 w_v(s'') = \left\{
\begin{array}{cc}
w_v(s')& \mbox{if $v \not= v'$,}\\
 w^+_{v'}& \mbox{if $v= v'$.}  
\end{array}
\right.
$$
We now show that $s''$ will still remain  a worst case scenario for $(\calP,X)$  satisfying conditions (1), (2) and (3).

Similar to  part (b),  since $w_v(s'') = w_v(s')$ for all $v \not\in V_j$ we have 
$$\forall i' \not=i,\quad\Theta(P_{i'},x_{i'}:s') =\Theta(P_{i'},x_{i'}:s'')
$$
and
$$
\forall j' \not=j,\quad \Theta(V_{j'}\cup\{x_i\},x_{i}:s') = \Theta(V_{\ell}\cup\{x_i\},x_{i}:s'').
$$

Now note that for all $v \in V_j$, the weight change implies 
$$W(V_j,x_i,v:s'') =
\left\{
\begin{array}{ll}
W(V_j,x_i,v:s') + ( w^+_{v'} - w_{v'}(s'))  & \mbox{if $d(v,x) \le d(v',x)$},\\
W(V_j,x_i,v:s') 					&  \mbox{if $d(v,x) > d(v',x)$}.
\end{array}
\right.
$$
Thus  $W(V_j,x_i,v^*:s'')=W(V_j,x_i,v^*:s')+ ( w^+_{v'} - w_{v'}(s'))$ and 
\begin{eqnarray*}
\Theta(V_j  \cup \{x_i\}, x_i: s'') &=&
\max_{ v \in V_j} 
\left(
d(v,x_i) +  \frac {W(V_j,x_i,v:s'')} c  
\right)\\
&=&d(v^*,x_i) +  \frac {W(V_j, x_i, v^* :s'')} c\\
&=&d(v^*,x_i) +  \frac {W(V_j, x_i, v^*:s')} c  + \frac {( w^+_{v'} - w_{v'}(s'))}{c} \\
&=& \Theta(V_j  \cup \{x_i\}, x_i: s') + \frac {( w^+_{v'} - w_{v'}(s'))}{c}.
\end{eqnarray*}


Then,  
$$\Theta(P_i,x_i:s'')  = \max_{1 \le \ell \le t} \Theta(V_{\ell} \cup\{x_i\},x_i:s'')= 
\Theta(P_i,x_i:s')+ \frac {( w^+_{v'} - w_{v'}(s'))}{c} $$
and
$$\Theta(\calP,X: s'') = \Theta(P_i,x_i:s'')  = \Theta(\calP,X: s') + \frac {( w^+_{v'} - w_{v'}(s'))}{c}.$$

From Lemma \ref {lem:dom3},
$$\Thetaopt(s'')\le  \Thetaopt(s') + \frac {( w^+_{v'} - w_{v'}(s'))}{c}.$$

Thus,
\begin{eqnarray*}\Theta(\calP,X: s'') - \Thetaopt(s'') 
 &=&\left( \Theta(\calP,X: s') +  \frac {( w^+_{v'} - w_{v'}(s'))} c \right)    - \Thetaopt(s'')\\
 &\ge&  \left(\Theta(\calP,X: s') + \frac {( w^+_{v'} - w_{v'}(s'))} c \right)    \\
 & & \hspace*{.5in} - \left(\Thetaopt(s')+  \frac {( w^+_{v'} - w_{v'}(s'))} c\right)\\
 &=&  \Theta(\calP,X: s')  - \Thetaopt(s') =\Rmax(\calP,X).
\end{eqnarray*}

The definition of $\Rmax(\calP,X)$ guarantees that  $\Theta(\calP,X: s'') - \Thetaopt(s'') \le \Rmax(\calP,X)$ and thus
$$\Rmax(\calP,X) = \Theta(P_{i},x_i:s'') - \Thetaopt(s').$$

Thus $s''$ is again  a worst case scenario for $(\calP,X)$  with  $P_i \in \calP$  a dominant subtree for $(\calP,X)$ under scenario $s''$, with $V_j$ a dominant branch in $P_i$.

Now, set $s'$ to be $s''$. 
As long as there exists  some  $v' \in V_j$ such that 
  $d(v',x_i) \ge  d(v^*, x_i)$ and  $w_{v'}(s') < w_v^+$,
the argument above can be repeated.  Since this process only increases  weights and never reduces  them, 
 it must terminate after a finite number of steps
 
Let $s_c$ be the scenario with which it terminates.
  By construction,  $s_c$ is  a worst case scenario for $(\calP,X)$ satisfying conditions (1), (2) and (3).

Furthermore, also by construction,  $s_c =s^*(D(V_j,x_i,v^*))$.  Thus, we have proven the claim with $s^* = s_c.$

%
%
%

\end{proof}

\medskip

\section{The Local Max-Regret Function}
\label{sec:local-max}

%

\begin{definition}
Let $T=(V,E)$ be a tree, $v \in V$, $T'=(V',E')$ a subtree of $T$ and $x \in V'.$ 
 Define the  {\em relative max-regret function} as
$$r(V',x) = \max_{s\in \calS^*(V', x)} 
\left\{
\Theta(V',x: s) - \Thetaopt(s)
\right\}
$$
Let $X \subseteq V$ and  $\calP\in \Lambda[X].$  Set
\begin{equation} r(\calP,X) =  
 \max_{1 \le i \le k} r(P_{i},x_i).
 \label{eq:MMregret}
\end{equation}
\end{definition}

Recall  that $|\calS^*(V',x)|  \le |V'|$.
This will permit efficiently calculating  $r(\calP,X)$. 
Surprisingly,  even though $r(\calP,X)$  is a locally defined function, it encodes enough information to fully calculate the  global value $\Rmax(\calP,X)$.

\bigskip

\begin{lemma} [$\emptyr$ is {\em almost} min-max monotone]
\label{lem:almostreg}\ \\
\begin{itemize}
\item[(a)] The function $\emptyr$  satisfies properties 1  and 3--6 of Definition \ref{def:MMM}.
\item[(b)] Let $X \subseteq V$ and  $\calP\in \Lambda[X].$    Then
$$ 
r(\calP,X)   
 = \Rmax(\calP,X)
$$
\item[(c)] \   
$$
\Ropt =  \min_{(\calP,X):  |X| = k}R_{\max}(\calP,X) =
         \min_{(\calP,X):  |X| = k} r(\calP,X)
$$
\end{itemize}
\end{lemma}

\begin{proof}

(a) For any fixed $s \in \calS$ set 
$$r'(V',x:s)             =  \Theta(V',x:s) - \Thetaopt(s).$$ 
For fixed scenario $s,$  $\emptytheta$ is minmax monotone.
Properties 1,3 and 4 all remain invariant under the subtraction of a constant, so $\emptyrp$  also satisfies properties 1, 3 and 4.
Since
$r(V',x) = \max_{s\in \calS^*(V',x)} r'(V',x:s)$ and properties 1, 3 and 4 also remain invariant under taking maximum,  $\emptyr$ also satisfies properties 1, 3 and 4.

\medskip

For  property 5, let $T' = (V',E')$ be a subtree of $T$, $x \in V'$  and $V_1,...,V_t$ the branches of $V'$ falling off of $x.$
We first claim that 
\begin{equation}
\label{eq:sim}
\forall \ell,\quad \max_{s \in \calS^*(V_\ell\cup\{x\},x)}  r'(V_\ell\cup \{x\},x:s) = 
\max_{s \in \calS^*(V',x)} r'(V_\ell \cup \{x\},x:s).
\end{equation}
Suppose not.  Then, for some $\ell$, there exists  $\ell' \not=\ell$  and   
$s^* \in \calS^*(V_{\ell'}\cup\{x\},x)$   such that 
\begin{equation}
\label{eq:rpcont}
  r'(V_\ell\cup \{x\},x:s^*) >
\max_{s \in  \calS^*(V_\ell\cup\{x\},x)} r'(V_\ell \cup \{x\},x:s).
\end{equation}

Since $\ell \not=\ell'$, for all $v' \in V_\ell,$  $w_{v'}(s^*) = w^-_{v'}  = w_{v'}(s^*(\emptyset))$ and thus
\begin{equation}\label{eq:rpcont2}
\Theta(V_\ell\cup \{x\},x:s^*(\emptyset)) = \Theta(V_\ell\cup \{x\},x:s^*).
\end{equation}
Furthermore,  by Lemma \ref{lem:dom2} (2), $\Thetaopt(s^*(\emptyset)) \le \Thetaopt(s^*).$ 
Since $s^*(\emptyset) \in \calS^*(V_\ell\cup\{x\},x),$ from Eq.~(\ref{eq:rpcont2}),
\begin{eqnarray*}
\max_{s \in  \calS^*(V_\ell\cup\{x\},x)} r'(V_\ell \cup \{x\},x:s) &\ge & r'(V_\ell\cup \{x\},x:s(\emptyset)) \\
 &=& \Theta(V_\ell\cup \{x\},x:s^*(\emptyset)) -\Thetaopt(s^*(\emptyset)\\
&\ge& \Theta(V_\ell\cup \{x\},x:s^*) -\Thetaopt(s^*)\\
&=&  r'(V_\ell\cup \{x\},x:s^*)
\end{eqnarray*}
contradicting (\ref{eq:rpcont}). Thus  (\ref{eq:sim}) is proved. Next note
\begin{eqnarray*}
r(V',x) &=& \max_{s \in \calS^*(V',x)} \left(  \Theta(V',x: s) - \Thetaopt(s)    \right)\\
          &=& \max_{s \in \calS^*(V',x)} \left(\max_{1\leq  \ell\leq t}   \Theta (V_\ell \cup \{x\},x:s)  - \Thetaopt(s) \right)\\
          &=& \max_{s \in \calS^*(V',x)} \left(\max_{1\leq  \ell\leq t} r'(V_\ell \cup \{x\},x:s). \right)\\
          &=&   \max_{1\leq  \ell\leq t} \left( \max_{s \in \calS^*(V',x)} r'(V_\ell \cup \{x\},x:s)\right)\\
        &=&   \max_{1\leq \ell \leq t} \left( \max_{s \in \calS^*(V_i\cup\{x\},x)} r'(V_\ell \cup \{x\},x:s)\right)   \hspace*{.2in} (\mbox{From Eq.~\ref{eq:sim}})\\
          &=&   \max_{1\leq  \ell\leq t} r\left(V_\ell\cup \{x\},x\right).
\end{eqnarray*}
$\emptyr$ therefore satisfies Property 5. Property 6 follows by definition.

\medskip

(b)
\begin{eqnarray*}
 \Rmax(\calP,X) &=& \max_{s \in \calS} \Bigl(\Theta(\calP,X: s)  - \Thetaopt(s)\Bigr) \\
 			&=&  \max_{s \in \calS} \left( \max_{1 \le i \le k} \Theta(P_{i},x_i: s)  - \Thetaopt(s) \right)\\
		         &=&  \max_{s \in \calS} \left( \max_{1 \le i \le k} \Bigl( \Theta(P_{i},x_i: s)  - \Thetaopt(s) \Bigr)\right)\\
 			&=&  \max_{s \in \calS} \left( \max_{1 \le i \le k}  r'(P_{i},x_i:s)\right)\\
 			&=& \max_{1 \le i \le k}    \left( \max_{s \in \calS} r'(P_{i},x_i:s)\right)\\
			&\ge& \max_{1 \le i \le k}    \left( \max_{s \in \calS^*(P_i,x)} r'(P_{i},x_i:s)\right)    \hspace*{.2in} (\mbox{Because $\calS^*(P_i,x) \subseteq \calS$  })\\
 			&=& \max_{1 \le i \le k} r(P_{i},x_i)  \\
                          & = & r(\calP,X)
\end{eqnarray*}

From  Lemma \ref{lem:wcr} we know there exists  a $s^* \in S^*(V',x)$  and $P_i \in \calP$  such that 
$$\Rmax(\calP,X)  = \Theta(P_i,x_i: s^*)  - \Thetaopt(s^*)$$
Thus
$$  \Rmax(\calP,X)   \le  r(P_i,x_i) \le r(\calP,X)$$
and $  \Rmax(\calP,X)  =r(\calP,X).$

(c)  Follows directly from (b).

\end{proof}

The previous lemma states that $\emptyr$ satisfies all of the properties of a  minmax monotone function EXCEPT for property 2.  Property 2 may be   violated since it is  quite possible that, for any particular $V'$,  that $r(V',x) < 0.$  As an example, suppose that $V'={x}$, a singleton node.  Since $\Theta(\{x\},x) =0$,
$$r(V',x) = \max_{s\in \calS^*(V',x)} 
\left\{
\Theta(V',x: s) - \Thetaopt(s)
\right\}
= -\max_{s\in \calS^*(V',x)}
\left\{
 \Thetaopt(s)
\right\}
$$
which other than in some special cases will be negative. Because of this $\emptyr$ is not minmax monotone and Theorem  \ref{Thm:Di} can't be directly applied. This can be easily patched, though.

\begin{lemma}
\label{lem:Final}
Let $T=(V,E)$ be a tree, $v \in V$ and $T'=(V',E')$ a subtree of $T$. 
Define the {\em local max-regret function} as 
$$\barr(V',x) = \max(r(V',x),\, 0).$$
Now let $X \subseteq V$ and  $\calP\in \Lambda[X].$  Set
$$ \barr(\calP,X) =  \max_{1 \le i \le k} \barr(P_{i},x_i).$$
Then
\begin{itemize}
\item [(a)]$ \emptybarr$ is a minmax monotone function.
\item[(b)] \  
$$
\Ropt =  
         \min_{(\calP,X):  |X| = k} \barr(\calP,X)
$$
Furthermore,  $R$ and $\barr$ have the same worst-case  evacuation protocols i.e.,  if $(\calP^*,X^*)$ are such that 
$$
 \min_{(\calP,X):  |X| = k}\barr(\calP,X) =\barr(\calP^*,X^*)
$$
then
$$\Ropt = R_{\max}(\calP^*,X^*).$$
\end{itemize}
\end{lemma}
\begin{proof}
(a) follows directly from Lemma  \ref{lem:almostreg} (a) and the definition of $ \emptybarr$.   (b) follows from Lemma \ref {lem:pos} and  Lemma  \ref{lem:almostreg} (b).
\end{proof}

\medskip

\section{The Algorithm}
\label{sec:alg}

Lemma \ref{lem:Final} shows that minmax regret can be expressed in terms of the {\em local} regret function $\emptybarr.$  
This section completes the proof of  Theorem \ref{thm:minmax} by  showing that plugging $\emptybarr$ into Theorem  \ref{Thm:Di} (from \cite{chengolin2018}) immediately yields the desired result.

\subsection{The Discrete Algorithm}
\label{subsec:DA}
In this subsection we continue assuming, as throughout  the paper until this point, that all sinks  must be located on vertices. 

Let  $T'=(V',E')$.   From Theorem \ref{thm: ksink},   $\Thetaopt(s)$ and $\Theta(V',x: s)$ can be calculated in 
$O( \max(k\log k,\log n)\,  k^2 n  \log^3 n)$ time
%
%
for any fixed  scenario $s$.  Recall that $|S^*(V',x)| \le  |V'|.$
Thus
$$\bar r (V',x)  = \max \Bigl(0, \max_{s\in \calS^*(V',x)} 
\left\{
\Theta(V',x: s) - \Thetaopt(s)
\right\}\Bigr)
$$
can be evaluated in 
$t_{\calA}(n') = O(n'   \max(k \log k,\log n)\,  k^2 n  \log^3 n)$ 
time where $n'= |V'|.$  

\medskip

Since this $t_{\calA}(n')$ is subadditive,  combining Lemma \ref{lem:Final} and Theorem  \ref{Thm:Di}   immediately implies
that  the minmax regret value can be calculated in 
$$O\Bigl( \max(k\log k,\log n)\,  k^2\,  t_{\calA}(n)\,  \log^2 n\Bigr)   = O\Bigl( \max(k^2 \log^2 k,\log ^2n)\,  k^4 n^2  \log^5 n\Bigr)$$ 
time, completing the proof of  Theorem \ref{thm:minmax}, assuming that sinks must located be on vertices.

\medskip

We complete this section by noting that the subadditivity is crucial.  To clarify, we note that the oracle works on  {\em subtrees} $T'$  of the fixed  full tree $T$ where  $|V'|=n'$ and $|V|=n$.
  Its $\tilde\Theta(n' n)$ running time (hiding the dependence  on $k$ and polylog terms in $\log n$) reflects the fact that, to calculate $\barr(V',x)$, it  is calculating the cost of $n'$ worst-case scenarios of the full tree $T$, each scenario requiring $\tilde\Theta(n)$ time.
  
\medskip

To find the tree partition, the algorithm of Theorem  \ref{Thm:Di}  calls the oracle $O(n)$ times on a large collection of overlapping subtrees. Some of those calls are to very large subtrees and some to very small ones.  
But, using subadditivity, it is able to prove that the {\em sum of the sizes} of all of those subtrees is $\Theta(n)$. Thus the total combined size of all worst case scenarios it actually constructs  is  $\Theta(n)$ and their total cost  $\tilde\Theta(n^2)$. Without subadditivity, this argument would fall apart and the algorithm of Theorem  \ref{Thm:Di} would have to assume that each of the $O(n)$ oracle calls would be to a subtree of size $\Omega(n),$ each requiring $\tilde\Theta(n^2)$ time, with total running  time being 
$\tilde\Theta(n^3).$ 

\subsection{The Continuous Algorithm}
This section permits loosening the problem constraints  to allow sinks to be located anywhere on an edge in addition to being on vertices.
See Fig.~\ref{fig: Fig6}.

\cite{chengolin2018} provides  an extension of Theorem \ref{Thm:Di} that is also applicable to these {\em Continuous} minmax monotone problems.

Some of the problem set up and definitions must then be naturally changed, e.g., in Definition \ref {def:Partitions}  and Properties 1-5 of  Section \ref{sec:minmax},
\begin{itemize}
\item $x \in  V'$ is replaced by $x \in T'$, i.e,  $x$  can  be a  vertex in $V'$  but it can also lie anywhere  on an edge in $E'.$
\item $X \subset V'$ is replaced by $X \subset T'.$
\item ``$x \not \in V'$ but $x$ a {\em neighbor} of $V'$'' is replaced by ``$x\not\in T'$  but there exists $u \in V'$, $v \not\in V'$ such that $(u,v) \in E$ and either $x = v$ or $x$ lies  in the interior of  the edge $(u,v)$''.
\item  Definition \ref{def:Partitions}(c) is extended so that if $x$ is internal to   edge $(u,v)$  then $x$   has exactly two branches $V_1,V_2$ falling off of it;  $V_1$  is the  subtree rooted at $u$ that does not contain $v$ and  $V_2$  is the subtree at $v$ that does not contain $u.$
\end{itemize}

\begin{figure}[t]
	\centering
\includegraphics[width=5in]{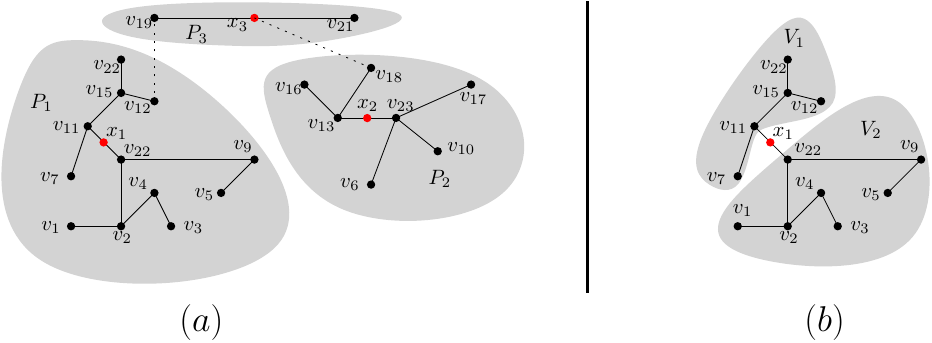}
\caption{An example of a solution of the continuous version of the problem.  In (a) the sinks $x_1$, $x_2$  and $x_3$ are on edges.  (b) further decomposes $P_1$ to  illustrate the fact that if a sink is on an edge $(u,v)$ then it has exactly two branches falling off of it.}
	\label{fig: Fig6}
\end{figure}

For consistency,  
the oracle $\calA$  extended to $x$ being on an edge must satisfy certain conditions.
 These are restated from \cite{chengolin2018} using the notation of this paper.
\begin{definition} (Fig.~\ref{fig:continuous})
\label{def:cont}
Let $T=(V,E)$ be  a tree and $\emptyf$ be a  minmax monotone cost function as defined at the beginning of Section \ref{sec:minmax}

For  $e=(u,v) \in E$,  orient $e$ so that it starts at $u$  and ends at $v$.  Let $V_u \subseteq V$ be a subtree of $T$ such that $u \in V_u$ but $ v \not\in V_u$ and  $x,x' \in e$.  Denote
\begin{eqnarray*}
x \le x'   &\quad \mbox{ if and only if} \quad  & \mbox{$x$ is on the path from $u$ to $x'$ }\\
x < x'   &\quad \mbox{ if and only if} \quad  & \mbox{ $x \le x'$  and  $x \not= x'.$}
\end{eqnarray*}
 $f(\cdot,\cdot)$ is  {\em continuous} if it satisfies:
\begin{enumerate}
\item  $f(V_u \cup \{x\}, x)$ is a continuous function in $\{x \,:\, u < x \le v\}.$
\item  $f (V_u\cup \{x\},x)$ is non-decreasing in $\{x \,:\, u \le x \le v\}$, i.e.,
$$\forall u \le x < x' \le v,\    f(V_u \cup \{x\}, x) \le  f (V_u(u)\cup \{x'\},x') .$$
\end{enumerate}
\end{definition}
Point 2 is the natural generalization of path-monotonicity.

\begin{figure}
	\centering
	\includegraphics[width=0.5\textwidth]{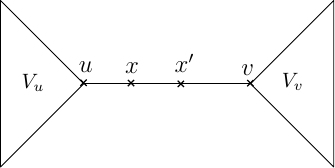}
	\caption{$V_u$, $V_v$ are, respectively, subtrees of $T$ rooted at $u$ and $v$.  Orient $(u,v)$ so that is starts at $u$ and ends at $v$.  Then $x < x'$. If the edge was oriented as ($v,u)$ then $x'< x'$  If $x \le x'$ then  $f(V_u \cup \{x\}, x) \le  f (V_v\cup \{x'\},x')$.}
		\label{fig:continuous}
\end{figure}

As  noted in \cite{chengolin2018},  $\emptytheta$ will naturally satisfy these conditions.
More specifically, let  $d(x,v)$ denote  the time required to travel from $x$ to $v$. It is natural to assume that this is a non-increasing continuous function in $x$ since flow travels smoothly without congestion {\em  inside} an edge. If the last flow  arrived at node $v$ at time $t$, then  it had arrived at  $x>u$ at time  $t-d(x,v)$. Thus 
\begin{equation}
\label{eq:dcontdef}
\Theta(V_u \cup \{x\},x:s) = f(V_u \cup \{v\},v:s) - d(x,v),
\end{equation}
so condition (1) is satisfied and condition (2) is satisfied for every $x$ except possibly $x=u.$
Now consider the time $t'$ that the last flow arrives at node $u$ and  let $t'+w$  be the time that this last flow {\em  enters} edge $(u,v)$.
Since flow doesn't encounter congestion inside  an edge, it arrives at $v$ at time   $t'+w + d(u,v).$ Then 
$$f(V_u,u) = t'  \le  t'+w =   (t'+w + d(u,v)) - d(u,v) =  \lim_{x \downarrow u}f(V_u \cup \{x\},x).$$ 
Thus condition (2) is also satisfied at $x=u.$  Note that $w >0$ only occurs if there is congestion at $(u,v)$ and this creates a left discontinuity, which is why 
 the range in condition (1) does  not include $x=u.$  

Since conditions (1) and (2) hold for every $\emptytheta$ they also hold for
\begin{equation}
\label{eq:ragain}
r(V',x) = \max_{s\in \calS^*(V',x)} 
\left\{
\Theta(V',x: s) - \Thetaopt(s)
\right\}
\end{equation}
and thus  for $\bar r(V',x) = \max(r(V',x),\, 0).$ 
That is,  $\bar r(V',x)$ is a contimuious minmax cost-function as defined by Definition \ref{def:cont}.

\begin{lemma}[\cite{chengolin2018}]
\label{lem:coneq}
Let $T'=(V',E')$ be  a tree, $\emptyf$ a  continuous monotone min-max cost function and $e=(u,v) \in E.$ 
 Let $V_u \subseteq V'$ be a subtree of $T$ such that $u \in V_u$ but $ v \not\in V_u$ and
 $V_v \subseteq V'$  a subtree of $T$ such that $v \in V_v$ but $ u \not\in V_v$. Finally, set $ \mathcal{T} \ge 0$ to be real.

Then both
\begin{equation}
\label{eq:st}
s_{\mathcal{T}} = \max_{x \in e} \Bigl( f(V_u\cup \{x\},x) \le \mathcal{T}\Bigr)
\end{equation}
and
\begin{equation}
\label{eq:alpha}
a:=\min_{x\in e} \max\Bigl(f(V_u \cup \{x\},x), f(V_v\cup \{x\},x)  \Bigr)
\end{equation}
exist.  
\end{lemma}

{\em Note: $s_{\mathcal{T}}$ and  $a$ will be needed by the algorithm in \cite{chengolin2018} to find  candidate sink locations.}

Finally, restated in our notation, it was shown 
\begin{theorem}[\cite{chengolin2018}]
\label{thm:contcg}
 Let $\emptyf$ be a continuous minmax monotone function with subadditive oracle $\calA.$
 Further suppose that the following two conditions are both satisfied: 
 \begin{itemize}
\item  $s_{\mathcal{T}}$ from   Lemma \ref{lem:coneq} along with the largest $x'\in e$ such that
 $$s_{\mathcal{T}} =   f(V_u\cup \{x'\},x') $$ 
 can be calculated using $O(1)$ oracle calls using
$O(t_{\mathcal{A}}(|V_u|))$ time.
 \item 
 $a$  from   Lemma \ref{lem:coneq} and any $x' \in e$ for which
 $$a=\max\Bigl(f(V_u \cup \{x'\},x'), f(V_v \cup \{x'\},x')  \Bigr)$$
 can be calculated using $O(1)$ oracle calls using
$O(t_{\mathcal{A}}( |V_u|+ |V_v|))$ time.
\end{itemize}
Then  the  continuous $k$-center partitioning problem on  $T$ can be solved in  time $$O\bigl( \max(k \log k,\log n) k^2 t_{\mathcal{A}}(n) \log^3 n \bigr).$$
	 \label{theorem:FastCC}
\end{theorem}
 Note that if $x \in e$, then combining Eqs.~(\ref{eq:dcontdef}) and (\ref{eq:ragain}) yields that 
 \begin{eqnarray*}
 r(V_u \cup \{x\},x) &=&r(V_u \cup \{v\},v) - d(x,v)\\
  r(V_v \cup \{x\},x) &=&r(V_v \cup \{u\},u) - d(u,x).
 \end{eqnarray*}

As noted in the analysis of the discrete case in  Section \ref{subsec:DA},  $r(V_u \cup \{v\},v)$ and $r(V_v \cup \{v\},u)$
can be evaluated, respectively,  in 
$t_{\calA}(|V_u|)$ and $t_{\calA}(|V_v|)$ time where 
$t_{\calA}(n') = O(n'   \max(k,\log n)\,  k n  \log^3 n)$.  Once  $r(V_u \cup \{v\},v)$ is known, $s_{\mathcal{T}}$ (and its associated $x'$) can be found in $O(1)$ time.  Similarly,  once $r(V_u \cup \{v\},v)$ and $r(V_v \cup \{v\},u)$ are known $a$  (and its associated $x'$) can also be found in $O(1)$ time.
Thus  $\emptybarr$ satisfies the conditions of Theorem \ref{thm:contcg} using the oracle from  Section \ref{subsec:DA}.

Similar to the discrete case,  combining   Theorem  \ref{thm:contcg} with 
  Lemma \ref{lem:Final}    immediately implies
that  the minmax regret value can be calculated in 
$$O\Bigl( \max(k\log k,\log n)\,  k^2\,  t_{\calA}(n)\,  \log^2 n\Bigr)   = O\Bigl( \max(k^2 \log^2 k,\log ^2n)\,  k^4 n^2  \log^5 n\Bigr)$$
time, completing the proof of Theorem \ref{thm:minmax} in the continuous case.

\section{Conclusions and Extensions}
\label{sec:conc}
This paper provided the first polynomial time algorithm for  the Minmax-Regret $k$-Sink Location problem on a Dynamic {\em Tree}  Network with uniform capacities and $k >1.$  It worked by noting (Section  \ref{sec:local-max})
that the minmax-regret function, which seems inherently global, can be expressed in terms of {\em local} minmax-regret functions and that (Section  \ref{sec:wcs})  each of these local min-max regret functions can be efficiently calculated.  It then applied  a tree-partitioning technique from \cite{chengolin2018} to these  local regret functions to calculate the global minmax-regret

One obvious extension would be to try and extend this result to the Minmax-Regret $k$-Sink Location problem on a Dynamic Tree  Network with {\em general} edge capacities.  As noted in the introduction, absolutely {\em no} results seem to be known for this general problem, even restricted to   $k=1$.  The structural reason for this is that, in the general capacity case, even though  Section  \ref{sec:wcs},   the expression of the global cost in terms of local costs, would still hold,  Section  \ref{sec:local-max}, the efficient calculation of these global costs, is not possible. More technically, the equivalent of Lemma \ref{lem:wcr} fails in the general capacity case in that it does not seem possible to restrict the set of worse case scenarios to a linear (or even polynomial) size set.
 Any extension of the approach in this paper to solving the general capacity problem would have to  first confront that difficulty.

Another extension would be to try to utilize the approach developed here to apply to other minmax-regret  functions.  This is possible.  As an example,  consider the weighted $k$-center problem. 
Recall that immediately after the statement of Theorem \ref{Thm:Di} we noted that the weighted $k$-center problem is modeled by setting
$$f(V,x) =  \max_{v \in V} w_v(s) d(x,v).$$
The minmax-regret weighted $k$-center problem can then be naturally defined.

It is straightforward to  modify all of the results in the paper, {\em including Lemma \ref{lem:wcr} and  Section  \ref{sec:local-max}},  to show that they all work for minmax-regret weighted $k$-center.  Plugging in the oracle costs would yield a final running time of
\begin{equation}
\label{eq:wkctime}
O\Bigl( \max(k^2 \log^2 k ,\log ^2n)\,  k^4 n^2  \log^4 n\Bigr).
\end{equation}
This is not particularly useful though because \cite{Averbakh1997} already gives a $O(n^2 \log^2 n \log \log n)$ time solution for the same problem.  Working through the details, the intuitive reason that  \cite{Averbakh1997}'s algorithm  is faster is because it strongly exploits the structural properties that, in the $k$-center problem, the cost of a subtree only depends upon
 {\em pairwise distances} between points and that the partition  $\cal P$ is uniquely determined  by $X$.
Theorem \ref{Thm:Di}  was designed to work with more general $\emptyf$ and therefore cannot take advantage of these problem specific properties.

For example, in  the sink-evacuation problem the cost of  a subtree is dependent upon interactions between all of the nodes in the subtree, e.g., via congestion effects and, unlike in the simple $k$-center problem,  a node might not be serviced by its {\em closest} sink.

%

As a final observation we note that the techniques in this paper {\em could} solve {\em generalized} versions of the minmax-regret weighted $k$-center problem that \cite{Averbakh1997}'s technique could not.  As a simple  example, suppose that we modify the $k$-center problem so that the cost of servicing node $v$ using center $x$ will be $d(x,v)$ times the sum of all vertex weights on the unique path connecting $x$ and $v$. More formally set 
$$W(x,v) =\{ u \,:\, \mbox{$u$ is on the path connecting $x$ and $v$}\},
$$
$$
w(x,v) = \sum_{u \in W(x,v)} w_u
$$
and
$$f(V,x) =  \max_{v \in V}  d(x,v) w(x,v).$$


Note that this $f(V,x)$  is a minmax monotone function which can be evaluated in $t_{\mathcal A}(|V|) = O(|V|)$ time.  The $k$-center partitioning problem and associated  minmax-regret weighted $k$-center partitioning problem for this generalized cost  function  are  then defined naturally. Exactly the same as described above,  plugging in the oracle would yield  the exact same  final running time as given in 
 Eq.~(\ref{eq:wkctime}) for this new  minmax regret $k$-center partitioning problem with this cost.  As noted, this is only an artificial problem constructed to illustrate the power of the technique developed in this paper. We are not aware of any currently outstanding  problems in the minmax-regret literature  for which this  paper's technique can improve the running time.

\medskip


\bibliographystyle{plain} 
\bibliography{main,Evacuation_Problems}


\end{document}